\documentclass[11pt]{article}

\usepackage{preamble}

\def\be{\begin{eqnarray}}
\def\ee{\end{eqnarray}}

\newcommand*{\email}[1]{\href{mailto:#1}{\nolinkurl{#1}} }

\begin{document}

\title{Fundamental theorem for quantum asset pricing}

\author{
Jinge Bao
\thanks{Centre for Quantum Technologies, National University of Singapore, Singapore 117543, \email{jbao@u.nus.edu}}
\and 
Patrick Rebentrost
\thanks{Centre for Quantum Technologies, National University of Singapore, Singapore 117543, \email{cqtfpr@nus.edu.sg}}
}

\date{\today}

\maketitle

\begin{abstract}
Quantum computers have the potential to provide an advantage for financial pricing problems by the use of quantum estimation. In a broader context, it is reasonable to ask about situations where the market and the assets traded on the market themselves have quantum properties.
In this work, we consider a financial setting where instead of by classical probabilities the market is described by a pure quantum state or, more generally, a quantum density operator. This setting naturally leads to a new asset class, which we call quantum assets. Under the assumption that such assets have a price and can be traded, we develop an extended definition of arbitrage to quantify gains without the corresponding risk. Our main result is a quantum version of the first fundamental theorem of asset pricing. If and only if there is no arbitrage, there exists a risk-free density operator under which all assets are martingales. This density operator is used for the pricing of quantum derivatives. To prove the theorem, we study the density operator version of the Radon-Nikodym measure change. We provide examples to illustrate the theory.
\end{abstract}

\section{Introduction}

The growth of capital markets has driven academic pursuits to describe these markets and to solve problems such as asset pricing \cite{Black1973,Merton1973} and portfolio optimization \cite{MR0103768,MR3235228}. 
Bachelier in his 1900 Ph.D. thesis was arguably the first to use random walks for analyzing speculation and financial options \cite{bachelier1900theorie,Davis2006Book}. In 1905, Einstein published his seminal work on diffusion which includes a description of the Brownian motion \cite{einstein1905molekularkinetischen}, which was rigorously defined by Wiener in 1923 \cite{wiener1923differential}.
Bachelier's ideas influenced later economists such as Samuelson \cite{samuelson1973mathematics,samuelson2009enjoyable,samuelson2015rational}.
The interplay of measure/probability theory and economics lead to the seminal Black-Scholes-Merton framework for fairly pricing European options under specific model assumptions \cite{Black1973,Merton1973}. 
By now there exists a substantial mathematical framework for arbitrage theory and pricing theory \cite{follmer2004stochastic}.
Finance and quantum technologies have recently been investigated in terms of quantum advantages that could arise from the use of quantum computers. 
Reviews are given in Refs.~\cite{orus2019quantum,bouland2020prospects,egger2020quantum, herman2022survey}. Problems considered are in portfolio optimization~\cite{barkoutsos2020improving,dasgupta2019quantum,hodson2019portfolio,rebentrost2018quantum,alcazar2020classical}, risk management~\cite{Woerner2018,alcazar2022quantum,han2022quantum}, and option pricing~\cite{Rebentrost2018finance,Martin2019,vazquez2021efficient,stamatopoulos2021towards,an2021quantum}.
The advent of quantum technologies, related to both computation and communication, entering financial markets motivates the question of how these technologies can transform the markets in fundamental ways.

In this work, we present a quantized study of finance where the market and assets themselves have quantum properties. We first propose a scenario where the financial market is described by a pure quantum state or a density operator instead of a probability vector and provide justification for such a scenario. We show that such a scenario naturally implies a definition of ``quantum assets'', which are introduced here. We study arbitrage in a mixed scenario when both classical and quantum assets are available. This discussion culminates in the first fundamental theorem of pricing quantum assets, which we prove here in analogy to the classical version. We define quantum derivatives in our context and show how the risk-neutral density operators can be used for their pricing. We study the quantum analog of the Radon-Nikodym measure change in the quantum density matrix picture. We consider two-level systems as examples to illustrate arbitrage and measure change.

\section{Quantum finance} \label{secQF}

The main starting point in this work is the assumption that the financial market is described by a quantum state or density operator in a Hilbert space, in contrast to a probability measure on a set of events.
We provide a potential scenario to justify this assumption next.
Then, in the following, we define the setting of this work in more detail. We note that for the classical analogue we heavily lean on Ref.~\cite{follmer2004stochastic}, see Appendix \ref{secPrelim}. 

\begin{figure}[t]
    \centering
\begin{centering}
\includegraphics[align=c,width=0.4\columnwidth]{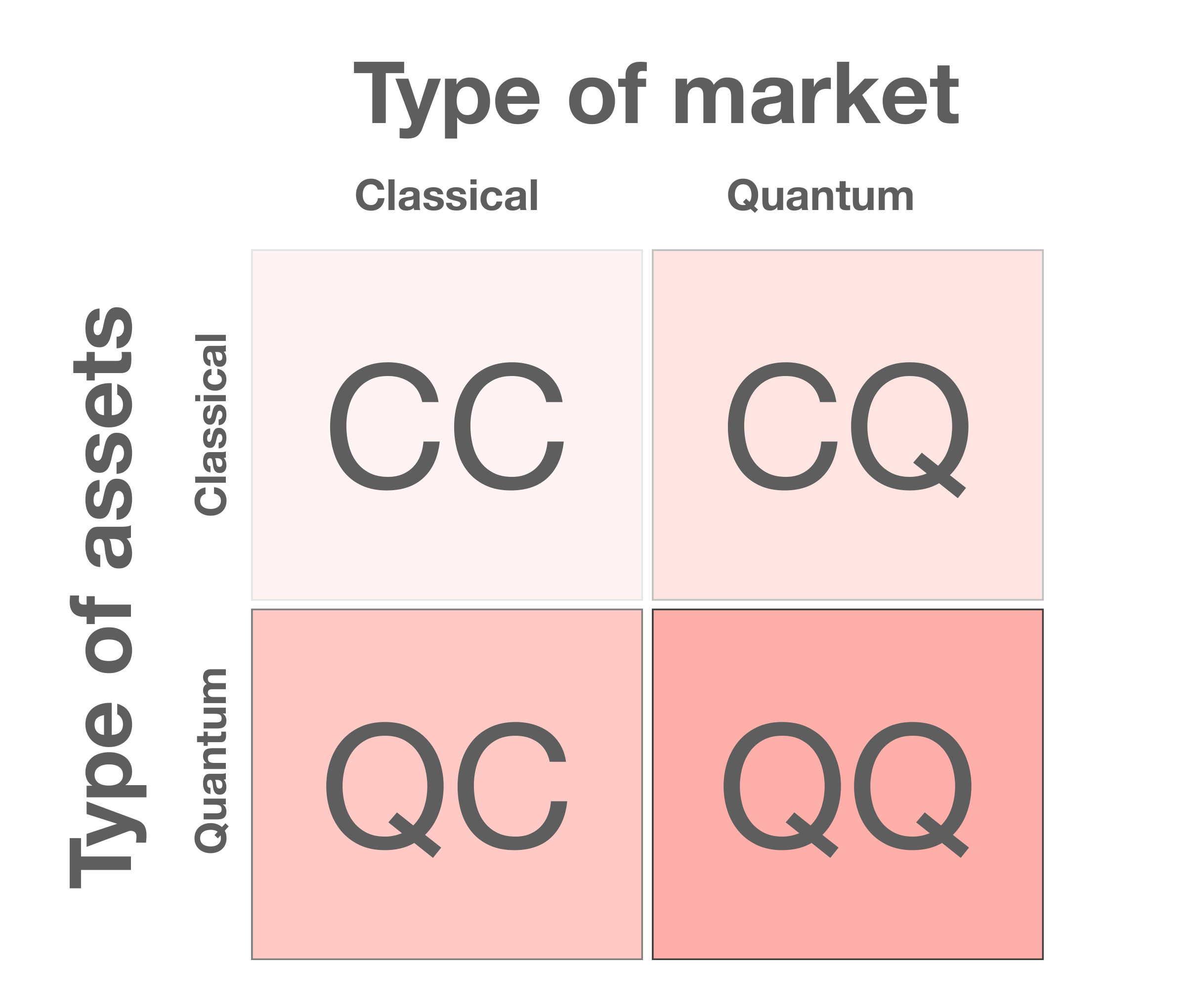}
\includegraphics[align=c,width=0.4\columnwidth]{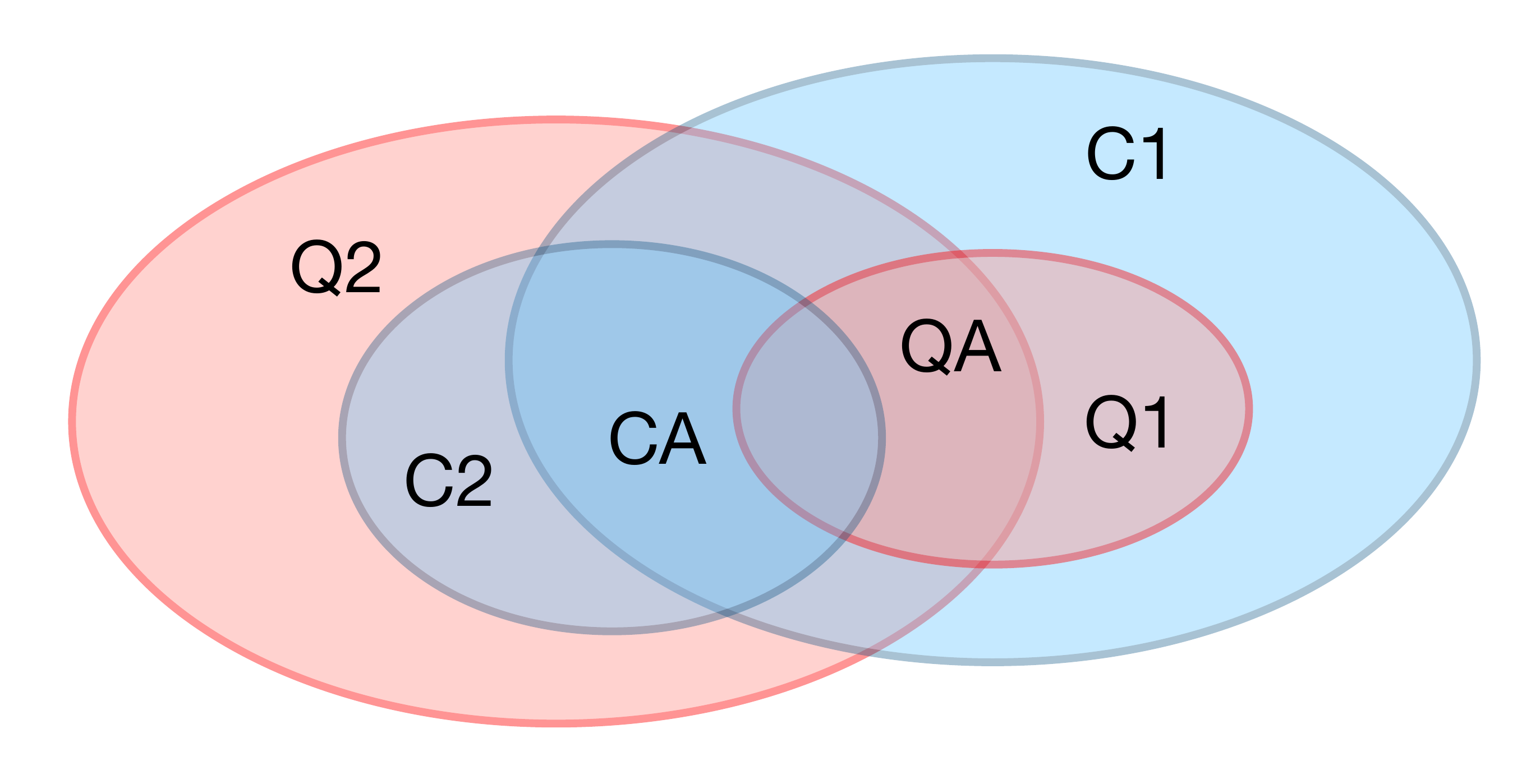}
\end{centering}
\caption{(Left panel) Different directions of quantum finance. Here we consider two dimensions: the type of market, which here is encapsulated by classical probabilities or quantum amplitudes, respectively, and the type of available and constructible assets. This figure is inspired by \cite{enwiki:1140135669,aimeur2006machine,dunjko2016quantum}, and we omit a third dimension to describe the type of information processing devices/algorithms (classical and quantum) used for solving financial problems. (Right panel) Illustration of the effect of the set of allowable states on the definition of quantum arbitrage, visualizing the direction ${\rm QC}\to {\rm QQ}$. In this figure, the classical scenario is the specialization of Def.~\ref{def:x_arbitrage_opportunity} to the case of only classical events, $\hat{\mathcal H} = \Omega$, while the quantum scenario is $\hat{\mathcal H} = \mathcal H^\Omega$.
The set of classical arbitrage portfolios $CA$ is the intersection of $C1$ and $C2$. The set of quantum arbitrage portfolios $QA$ is the intersection of $Q1$ and $Q2$, see Appendix \ref{appArb}.}
    \label{figQF}
    \label{figArb}
\end{figure}

\begin{figure}[t]
    \centering
\includegraphics[width=1.0\textwidth]{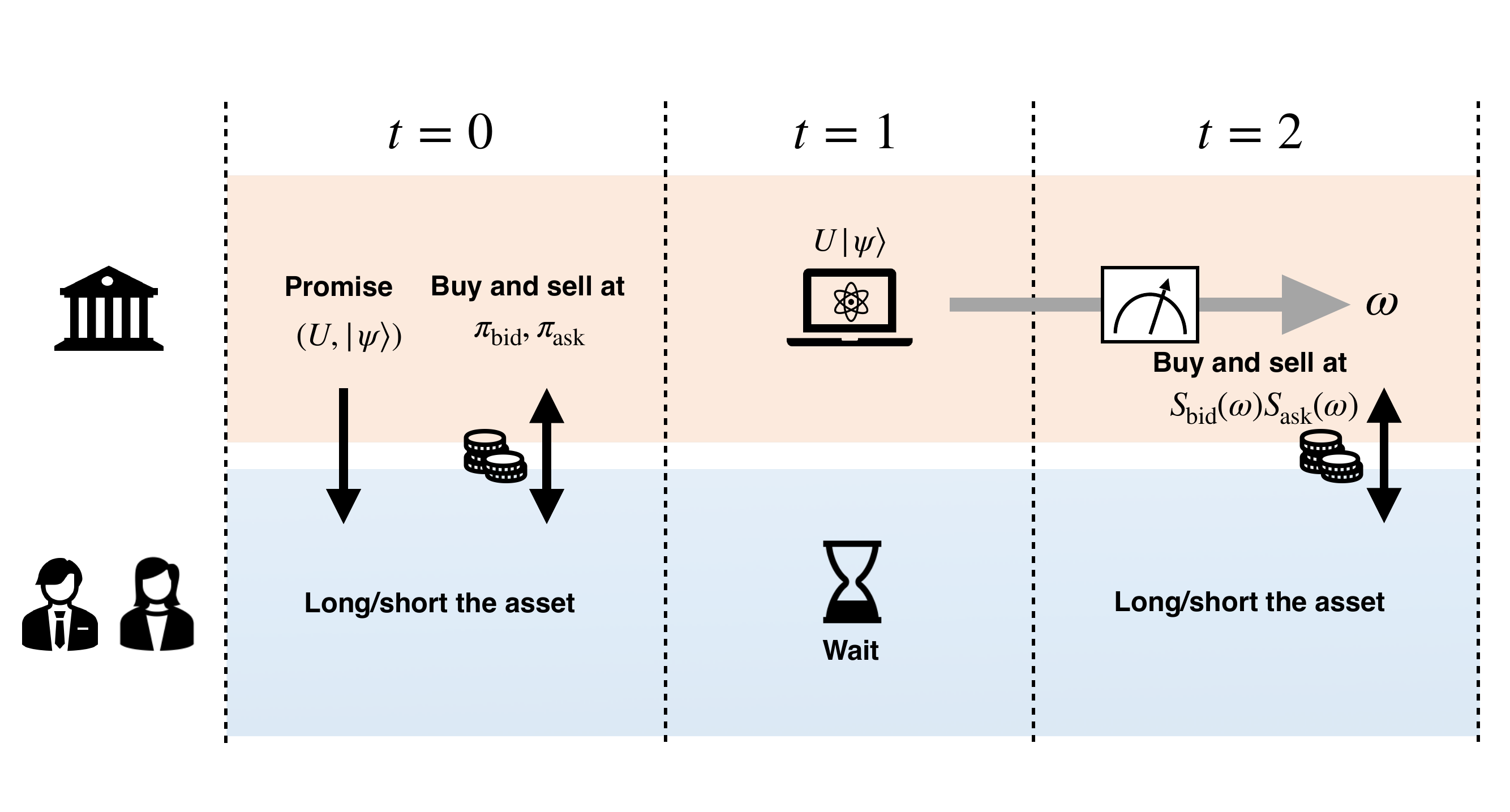}
\caption{A toy scenario. The top panels consider the Market Maker (MM) and the bottom panels consider an investor. The black arrows show interactions between MM and the investor. At $t=0$, the MM announces the unitary $U$, quantum state $\ket \psi$, initial price $\pi$ and functions $S_{bid}$, $S_{ask}: \Omega \mapsto \mathbbm R_{+}$. At $t=1$, the MM prepares the quantum state  $U\ket \psi$. At $t=2$, the state is measured and the MM announces the ask/bid price based on the measurement outcome $\omega$.}
    \label{figQScenario}
\end{figure}

\paragraph{Toy scenario.}
Here, we consider a scenario that would justify the quantum state description of a financial market. Consider a large financial institution with the financial might to significantly influence asset prices or create a liquid market for them. We call this institution a \textit{market maker (MM)}. 
Market-making \cite{radcliffe1990investment} is an important part of finance and means that buy/sell prices are quoted for some asset (say, a commodity, a stock, or a complex financial derivative) by the market maker or liquidity provider.
The market maker hopes to make a profit on the difference between the \textit{bid price} and the \textit{ask price}, which is known as the  \textit{bid-ask spread}.

Consider a market maker with access to a quantum computer. Consider also a benchmark asset that can be traded today $t=0$ at current prices $\pi_{\rm bid}$ and 
$\pi_{\rm ask}$.
A possible scenario would consist of three timesteps as follows.
\begin{itemize}
\item
At $t=0$, the market maker announces that they will, at $t=2$, set their bid price ($S_{\rm bid}$) and ask price ($S_{\rm ask}$) of the benchmark asset. At $t=0$, the MM also announces an $n$-qubit initial state $\ket \psi$ and a unitary $U \in \mathbbm C^{N\times N}$, with $N=2^n$. The MM further guarantees that both prices will be determined by the output of the quantum computation related to $(\ket \psi,U)$, for example, a measurement in the computational basis. The price determination will be via the known functions $S_{\rm bid}(\omega)$ and $S_{\rm ask}(\omega)$, where the measurement result is denoted by $\omega \in \{0,1\}^n$, where the probability of obtaining the result $\omega$ is $\vert \bra \omega U \ket \psi \vert^2$. 
\item 
After the announcement, at the time $t=1$, the MM sets up the quantum computation and  generates the quantum state $ U \ket \psi$. At this time the quantum state has not been measured yet.
\item At the next time $t=2$, the MM then performs the measurement of the quantum state in terms of the computational basis and acts as promised before.
\end{itemize}
Importantly, at the time $t=1$, the probabilities of the measurement are described by the quantum state, which can be in a superposition of the measurement outcomes. 

Now, consider the perspective of an investor in such a scenario. The investor receives the promise from the MM at $t=0$ and the knowledge of the scenario including ($\ket \psi, U$) and the functions $S_{\rm bid}(\omega)$ and $S_{\rm ask}(\omega)$ for $t=2$.
The investor has the chance to invest (buy/sell) in the benchmark asset at today's prices $\pi_{\rm bid}$ and $\pi_{\rm ask}$.
From the investor's point of view, at time $t=1$ the probabilities of the market for the benchmark asset will be encoded in the state $\rho = U\ket \psi \bra \psi U^\dagger$. 
The investor has to reconsider a host of fundamental financial questions. 
One question is about new asset classes that depend on the benchmark asset and that 
can take advantage of the quantum state at time $t=1$ or allow to hedge risks in this scenario.
Another question is about the expected value of a portfolio that includes such assets viewed at time $t=0$. 
Finally, the investor may reason about the issue of arbitrage, i.e., the possibility of gains without taking corresponding risks, and the pricing of hitherto unpriced assets. 
In the remainder of this work, we assume that  bid and ask prices for the assets traded in the market are the same for simplicity.

While the preceding discussion illustrates a potential toy scenario, we now define the setting in more detail.

\paragraph{Sample space and Hilbert space.}
In probability theory, a finite sample space is given by a set $\Omega=\{\omega_1,\cdots,\omega_K\}$ of primary events.
In the quantum case, we associate the classical finite sample space $\Omega$ with the basis of a Hilbert space.  
Given $\Omega$, each $\omega_i \in \Omega$ is associated with a basis vector $\ket {\omega_i}$, and the basis vectors form an orthonormal basis.
A Hilbert space is defined by $\mathcal H^\Omega:={\rm span}\{\ket {\omega_1},\dots,\ket {\omega_K}\}$. A vector $\ket \psi \in \mathcal H^\Omega$ has the representation $\ket \psi = \sum_{i=1}^K \alpha_i \ket {\omega_i}$ with complex coefficients $\alpha_i=\braket{\omega_i}{\psi}$.
In general, we do not necessarily require this association with the classical $\Omega$ and can start from the Hilbert space immediately, in which case $\Omega$ can be thought of as some basis of $\mathcal H$, such as the computational basis of a register of qubits. For the arbitrage discussion,  we define a subset $\hat{\mathcal H }\subseteq \mathcal H^\Omega$ of allowable states, see Section \ref{secArb}.

\paragraph{Market density operator.}
A probability measure is a measure (on a $\sigma$-algebra of $\Omega$) with positivity and normalization constraints and is usually denoted by $\mathbbm P$. 
The finite setting implies for $\mathbbm P$ to be the probability vector $\qvec p \in \mathbbm R^K$,  where $p_i = \mathbbm P(w_i)$ for all $w_i \in \Omega$ and $i \in [K]$. 
We generalize the classical probability vector to a density operator $\rho \in \mathbbm C^{K \times K}$. A density operator is a positive semi-definite self-adjoint operator with unit trace. 
We define the set of density operators in $\mathcal H^\Omega$ to be $\mathcal D(\mathcal H^\Omega)$.
In general, the density operator takes into account quantum correlations between events in the sample space.    
A density operator is a correct description if there is some process that leads to quantum superpositions of events in the time step $t=1$. As an extremely simplified case, consider the situation where a market maker has a single quantum coin which at $t=1$ is prepared in the state $\ket \psi = \frac{1}{\sqrt{2}}\left( \ket 0 + \ket 1 \right)$, a uniform superposition of the two outcomes $\ket 0$ and $\ket 1$. 
The corresponding density matrix is the rank $1$ matrix $\rho = \ket \psi \bra \psi$ of all entries being $1/2$ in the computational basis.
Let the time $t=1$ be before the quantum coin is measured. Despite the coin not being measured at $t=1$, an investor may want to reason about the properties of a portfolio and the value of a portfolio at time $t=1$ viewed from the time $t=0$. 

A technical concept appearing not  only in financial theory is the change from one probability measure to another probability measure, called Radon-Nikodym measure change. As it will also be important for the arbitrage discussion in our context, Appendix \ref{secRN} discusses a version of the Radon-Nikodym derivative for density operators. 
We comment on the nullspace of $\rho$. In the finite classical setting, it is often assumed that all the classical events $\omega \in \Omega$ have strictly positive probability. In the quantum setting, we allow the dimension of the nullspace of $\rho$ to be greater than zero. There can exist $\ket \psi \in \mathcal H^\Omega$ for which $\rho \ket \psi = 0$. We denote the projector into the nullspace as $\Pi_{\rm null}(\rho)$ and the projector into the positive-eigenvalue space  as $\Pi_{+}(\rho) = \mathbbm 1 - \Pi_{\rm null}(\rho)$. 

\paragraph{Quantum assets.}
If we accept the assumption that a quantum state or density operator describes the underlying market probabilities, it is reasonable to ask about new asset classes that may arise from this assumption. We here define an asset class which is a generalization of the classical asset class. A possible name is ``Quantum Hamiltonian Asset'', and for the purpose of this work, it is  called Quantum Asset. 
\begin{definition}[Quantum Asset]
\label{def:quantum_asset}
Given a Hilbert space $\mathcal H$, with ${\rm dim} \mathcal H=K$,
a \textit{quantum asset} is a positive semidefinite Hermitian matrix $\qmat S \in \mathbbm C^{K \times K}$. For $\ket \psi \in \mathcal H$, the future payoff is $\bra \psi \qmat S \ket \psi \geq 0$. The asset can be bought and sold for some price today, where the price is denoted by $\pi > 0$.
\end{definition}
The assets have a price for which they can be bought and sold, and in that sense, the value of these assets is defined. 
This definition leaves questions about the existence/validity of such assets and their intrinsic value for future work. 
Given a classical asset $\qvec S_{\rm cl} \in \RR^K$, its quantum asset embedding is given by the diagonal matrix $\qmat S = {\rm diag}(\qvec S_{\rm cl}) \in \RR^{K \times K}$. 
A quantum asset is diagonal in its eigenbasis, which motivates a financial interpretation. Each eigenstate can be considered a natural event for the quantum asset, in contrast to the primary events in $\Omega$. Each eigenvalue is the outcome or payoff of this asset when the corresponding event happens.

As a simple example of the $K=2$ case, consider a quantum asset $\qmat S$ in the Pauli basis as $\qmat S = a \qmat I + b \qmat \sigma_x + c \qmat \sigma_y + d \qmat \sigma_z$, with real coefficients $a$, $b$, $c$, and $d$. If we let $b=c=d=0$ and $1+r =a>0$, $\qmat S$ becomes a \textit{classical risk-free} asset with rate $r$, whose payoff is always $1+r$. When $a>d>0$ and $b=c=0$, $\qmat S$ is a \textit{classical risky} asset whose payoff is either $a+d$ or $a-d$ that is determined by the sign and amount of $d$. When $b$ or/and $c$ are nonzero, $\qmat S$ can be thought of as a genuine quantum asset. 

\paragraph{Expected value of quantum assets.}
Given a density operator $\rho$ and a single quantum asset $\qmat S$, we define the expected value of the asset under $\rho$ at time $t=1$ viewed at time $t=0$ as
\be
\mathbbm E^{\rho}[\qmat S]:={\rm tr} \left\{\rho \qmat S \right \}.
\ee
Let $\qmat S$ have the spectral decomposition $\qmat S=U\Lambda U^\dag$. Thus, $\mathbbm E^{\rho}[\qmat S]={\rm tr} \left\{\rho \qmat S\right\} = {\rm tr}\left\{\rho U \Lambda U^\dag\right\}={\rm tr}\left\{U^\dag \rho U \Lambda\right\}=\mathbbm E^{\rho'}[\Lambda]$, where we define $\rho':=U^\dag \rho U$. Here,  $\rho'$ is also a density operator, and the quantum asset $\qmat S$ is changed to a diagonal asset. In other words, the quantum asset can be regarded as a classical asset under a certain density operator. However, simultaneous diagonalization of multiple quantum assets is usually not possible.

Consider a simple example where $\rho  = \frac{1}{2} (\mathbbm 1 + \Delta \sigma_x)$, where $\Delta \in[-1,1]$ is the amount of coherence in the density matrix. In addition, consider the quantum asset $\qmat S = (\mathbbm 1 + q \sigma_x)$, where $q \in[-1,1]$ can be seen as the amount of ``quantum-ness'' in the asset. A unit amount of the asset will have the expected value 
$\mathbbm E^\rho[\qmat S] = 1 + q\Delta$. Depending on the signs of $\Delta$ and $q$, we obtain a lower or higher expectation value than the corresponding classical value of $1$ at $\Delta =0$ or $q=0$.

\paragraph{Main setting.}
The main setting in this work is a single-period model of today ($t=0$) and tomorrow ($t=1$). We can invest in one risk-free asset with the fixed interest rate $r$ and today's price of $\pi_0=1$, and $d$ quantum assets or embedded classical assets. We denote the collection of assets (tomorrow's payoffs) as $\overline{\qten S}=(\qmat S_0,\qmat S_1,...,\qmat S_d)$, which formally can be regarded as a $(d+1) \times K \times K$ tensor, while conceptually all assets live in the same Hilbert space. Together with the price vector $\overline {\qvec \pi} \in \mathbbm R_+^{d+1}$, we call the tuple $(\overline {\qvec \pi}, \overline{\qten S})$ a \textit{price system}.
The definition of a portfolio $\overline {\qvec \xi} \in \mathbbm R^{d+1}$ is a vector of the amount of assets being held by an investor, the same as in the classical setting. This setting allows for negative entries in the portfolio which corresponds to taking out a loan for the risk-free asset or short-selling a risky asset, respectively.
Given a portfolio $\overline {\qvec \xi}=(\xi^0,\qvec \xi) \in \mathbbm R^{d+1}$, the expected value the portfolio is given by
\be
\mathbbm E^\rho[\overline {\qvec \xi} \cdot \overline {\qten S}]
= {\rm tr}\left \{\rho\  \overline {\qvec \xi} \cdot \overline {\qten S}\right\}.
\ee
With the main setting, the quantum probabilities, and the quantum assets being defined, we can now discuss arbitrage and risk-neutral density operators.

\section{Arbitrage theory for quantum assets} \label{secArb}

The next fundamental question is the relationship between quantum assets and arbitrage opportunities.
In analogy to the classical definition from \cite{follmer2004stochastic}, we give our definition of a \textit{quantum arbitrage opportunity}. 
To generalize the situation, we also allow the use of a subset $\hat{\mathcal H }\subseteq \mathcal H^\Omega$ of allowable states.
First, due to constraints on the scenario, some sets of states may never appear, and we would like to have the freedom to include these situations. Note that we cannot model this situation with the null space of the density operator. 
Second, if we hold the price system and market density operator constant, we can consider the classical limit of this definition of arbitrage if we take $\hat{\mathcal H }=\Omega$.
The following definition only depends on the allowable  states in $\hat{ \mathcal H}$ and the nullspace of the density matrix $\rho$.
\begin{definition}[Quantum arbitrage opportunity]
\label{def:x_arbitrage_opportunity}
Let $(\overline {\qvec \pi}, \overline{\qten S})$ be a price system, $\hat {\mathcal H}\subseteq \mathcal H^\Omega$ be a set of quantum states and $\rho \in \mathcal D(\mathcal H^\Omega)$ be the market density operator. 
A portfolio $\overline {\qvec \xi} \in \RR^{d+1}$ is called a \textit{quantum arbitrage opportunity} if
todays value is $ \overline {\qvec \xi} \cdot \overline {\qvec \pi}\leq 0 $, and
for all $\ket \psi \in  \hat{\mathcal H}$ with $\bra \psi \rho \ket \psi > 0$, it holds that $\bra \psi \overline {\qvec \xi} \cdot \overline {\qten S} \ket \psi \geq 0 $ and
there exists at least one $\ket \psi \in \hat{\mathcal H}$ with $\bra \psi \rho \ket \psi > 0$ and $ \bra \psi \overline {\qvec \xi} \cdot \overline {\qten S}) \ket \psi > 0 $. 
\end{definition}
This definition specializes to the definition of classical arbitrage (Definition \ref{def:arbitrage_opportunity}) when all the assets are classical (diagonal) and the states are only chosen from $\ket \psi \in \Omega$.
In Figure \ref{figArb} and Appendix \ref{appArb}, we illustrate the relationship between classical limit $\hat{\mathcal H }=\Omega$ and quantum arbitrage $\hat{\mathcal H }=\mathcal H^\Omega$, when price system and density operator are fixed. We see an interesting separation of the different regimes, which allows a variety of scenarios of classical and quantum arbitrage. 

An arbitrage opportunity can be characterized in terms of only the risky assets and the market rate.
Lemma \ref{lem:x_arbitrage_risky_asset}, allows relating arbitrage to returns beyond the market rate $r$ from the riskless asset. It is shown that 
the market model admits a quantum arbitrage opportunity if and only if there is a vector $\qvec \xi \in \mathbbm R^d$ such that
for all $\ket \psi \in \hat{ \mathcal H}$ with $\bra \psi \rho \ket \psi > 0$, we have that $ \bra \psi \qvec \xi \cdot \qten S \ket \psi \geq (1+r) \qvec \xi \cdot \qvec \pi$ and 
there exists at least one $\ket \psi \in \hat{ \mathcal H}$ with $\bra \psi \rho \ket \psi > 0$, such that $\bra \psi \qvec \xi \cdot \qten S \ket \psi > (1+r) \qvec \xi \cdot \qvec \pi.$
See Lemma \ref{lem:x_arbitrage_risky_asset} for the proof.
Based on this lemma it is convenient to define $\qten Y=(\qmat Y_1,\dots,\qmat Y_d)$ where $\qmat Y_i \coloneqq \frac{\qmat S_i}{1+r}-\pi_i \mathbbm I$ as the \textit{discounted net gains}. 
It is also important to investigate those market models which do \textit{not} admit any arbitrage opportunity. The models with this property are named \textit{quantum arbitrage-free market models}. 
In addition, Corollary \ref{cor:x_arbitrage_discounted_net_gains} shows that the absence of arbitrage is equivalent to the property that risk-free non-negative net gains have to be zero net gains.
 
We illustrate the Definition \ref{def:x_arbitrage_opportunity} and Lemma \ref{lem:x_arbitrage_risky_asset} with simple examples. We show that quantum effects can either turn an arbitrage-free model into an arbitrage model, or turn an arbitrage model into an arbitrage-free model, or do not affect an arbitrage-free model.

Consider the two base events $\{ \omega_1, \omega_2\}$ and the density operator $\rho  = \frac{1}{2} (\mathbbm 1 + \Delta \sigma_x)$, where $\Delta \in[-1,1]$. 
Both base events are not in the null-space of $\rho$, since $\bra{\omega_1}\rho \ket {\omega_1}=\bra{\omega_2}\rho \ket {\omega_2}=1/2$.
Let us be given the risk-free asset, with interest rate $r$, and 
the asset $\qmat S = \frac{a+b}{2} \mathbbm 1 + \frac{a-b}{2}\sigma_z + q \sigma_x$, with today's price $\pi> 0$  and $0<a <b$ and $q\in [-\sqrt{ab},\sqrt {ab}]$.
The range for $q$ comes from the semi-positivity constraint for the asset.
The payoffs for the classical events evaluate as 
$\bra{\omega_1} \qmat S \ket {\omega_1} = a$ and $\bra{\omega_2} \qmat S \ket {\omega_2} = b$.
If $a <  \pi (1+r) < b$, then
$\bra{\omega_1} \qmat S \ket {\omega_1} < (1+r)\pi$ and $\bra{\omega_2} \qmat S \ket {\omega_2} > (1+r)\pi$.
We find whether there is arbitrage on the classical events by Lemma \ref{lem:x_arbitrage_risky_asset}.
Since for all $\xi \in \mathbbm R$ the portfolio $\xi \qmat S $ has the possibility of a payoff smaller than $(1+r) \pi$, the setting is classically arbitrage-free. The classical result is independent of the choice of $\Delta$ and $q$.
Somewhat surprisingly, the setting is also quantumly arbitrage-free, because we have the event $\bra{\omega_1} \qmat S \ket {\omega_1}<(1+r) \pi$, which gives the chance to lose money compared to the risk-less investment. 

Consider the asset $\mathcal S$ as in the previous example specialized for the case of today's price $\pi=1$ and $a=b=1+r$. Evaluate the classical events
again as 
$\bra{\omega_1} \qmat S \ket {\omega_1} = 1+r$
and $\bra{\omega_2} \qmat S \ket {\omega_2} = 1+r$.
Hence, there is again no classical arbitrage.
However, let $\ket \psi = \psi_+ \ket + + \psi_- \ket -$, with $\psi_+,\psi_-\in \mathbbm C$ such that $\vert \psi_+\vert^2 +\vert \psi_-\vert^2 =1$, and consider the set of states 
$\hat{ \mathcal H} := \{ \ket \psi : \vert \psi_+\vert^2 \geq \vert \psi_-\vert^2\}$.
We have
$\bra{\psi} \qmat S \ket {\psi}= 1+r + (\vert \psi_+\vert^2- \vert \psi_-\vert^2) q.$
For $q>0$, we obtain arbitrage, because when $\vert \psi_+\vert^2 > \vert \psi_-\vert^2$ we have strictly positive payoffs and there are never strictly negative payoffs.

Consider again the asset $\mathcal S$  specialized for the case of $\pi=1$ and $1+r< a < b$. We now have arbitrage on the classical events since $1+r <\bra{\omega_1} \qmat S \ket {\omega_1} <\bra{\omega_2} \qmat S \ket {\omega_2}$ and the risky asset always returns more than the risk-free rate. 
In the quantum setting, let the state space 
be $\hat{ \mathcal H} = \mathcal H^\Omega$. 
Note that $\mathcal S - (1+r)\mathbbm 1$ is a positive semi-definite matrix if $q \in [-q_0,q_0]$, where $q_0 = \sqrt{(a-1-r)(b-1-r)}$.
Hence, 
for all $q\in [-\sqrt{ab}, -q_0] \cup [q_0,\sqrt {ab}]$, we obtain some $\ket \psi$ for which $\bra{\psi} \qmat S \ket {\psi} < 1+r$.
Hence, we have the case of no quantum arbitrage, as there are quantum events $\ket \psi$ such that the asset has a chance for a loss compared to the risk-free asset.
The use of $\hat{\mathcal H}$ gave us a more general definition of quantum arbitrage and hence more interesting cases. However, for the remainder of this work, we consider $\hat{\mathcal H} = \mathcal H^\Omega$.

In analogy to the classical case, it makes sense to define a \textit{risk-neutral density operator} or \textit{martingale density operator}, which is a natural generalization of the risk-neutral probability measure.
Under such a density operator the discounted expected value of a quantum asset is exactly its price today.
Such density operators will be the content of the first fundamental theorem of the pricing of quantum assets.
A density operator $\rho^\ast \in \mathbbm C^{K \times K}$ of $\Omega$ is called a \textit{risk-neutral density operator}, or a \textit{martingale density operator}, if
\be \label{eqMartingaleDensityOperator}
\pi_i = {\rm tr} \left \{ \rho^\ast \frac{\qmat S_i}{1+r}\right\},
\ee
for all $i \in [d]_0$.
A natural question arises about changing the original density operator $\rho$ into the risk-neutral density operator $\rho^\ast$, which is the motivation for Appendix \ref{secRN}. 

We end this section with the following example illustrating Eq.~(\ref{eqMartingaleDensityOperator}).
Consider the asset $\qmat S = \frac{a+b}{2} \mathbbm 1 + \frac{a-b}{2}\sigma_z + q \sigma_x$ with price $\pi$, where $0 < a < \pi (1+r) < b$ and $q \in[-\sqrt {ab},\sqrt {ab}]$.
Consider a market density operator $\rho^\ast \coloneqq p \ket 0 \bra 0 + (1- p) \ket 1 \bra 1 + \Delta \sigma_x$,
with $p = (b-(1+r)\pi)/(b-a)$ and $\Delta \in[-\sqrt {p(1-p)},\sqrt{p(1-p)}]$.
Checking the risk-neutrality property amounts to
$
\mathbbm E^{\rho^\ast}\left [\frac{\qmat S}{1+r}\right] = \pi + \frac{2 q\Delta}{1+r}
$.
If $q=0$ or $\Delta = 0$, we have the desired property. In other cases, this example suggests that the market prices for the quantum asset are not arbitrage-free and should be modified to 
$\pi + \frac{2q\Delta}{1+r}$.

\section{First fundamental theorem and quantum derivatives} \label{secFT}

In this section, we present the result of the existence of risk-neutral density operators. The result and proof are analogous to the classical case presented in \cite{follmer2004stochastic}.
The risk-neutral density operator from Eq.~(\ref{eqMartingaleDensityOperator}) may not be unique. The set of risk-neutral density operators which are equivalent to $\rho$ is denoted by $\mathcal P$, i.e., 
$\mathcal P := \{\rho^\ast|\rho^\ast \textrm{ is a risk-neutral density operator for which } \rho^\ast \approx \rho \}.$
A version of the Radon-Nikodym derivative, discussed in Appendix \ref{secRN}, allows changing between density operators.
Our main theorem characterizes the arbitrage-free model in terms of the set $\mathcal P$.

\begin{theorem} 
[Fundamental theorem of quantum asset pricing (FTQAP)]\label{thmMain}
A market model is \textit{quantum arbitrage-free} if and only if $\mathcal P \neq \emptyset$. In this case, there exists a $\rho^\ast \in \mathcal P$.
\end{theorem}
We show the proof in Appendix \ref{appProofMain}.
Given the primary quantum assets with arbitrage-free prices, it is natural to ask about other financial assets. A \textit{quantum derivative} here is a quantum asset without a current price. 
A quantum derivative is a positive semidefinite Hermitian matrix $\qmat V \in \mathbbm C^{K \times K}$. For $\ket \psi \in \mathcal H^\Omega$, the future payoff is $\bra \psi \qmat V \ket \psi \geq 0$. 
The derivative does not have a current price.
This definition includes derivatives which are some function of the primary assets. For example, given a primary asset $\mathcal S$ and some function $v : \mathbbm R \to \mathbbm R$, a derivative can be defined as $\mathcal V \coloneqq v(\mathcal S)$, i.e., the matrix where the function is applied to the eigenvalues of the primary asset.
Valuation of quantum derivatives is assigning them an arbitrage-free price. The main result from the fundamental theorem (Theorem \ref{thmMain}) is that in the absence of arbitrage, we obtain a set of risk-free density operators.  Let a density operator be such that
$
\rho^\ast \in \mathcal P.
$
Any density operator in $\mathcal P$ can be used for arbitrage-free price determination.
We are given a quantum derivative $\qmat V \in \mathbbm C^{K \times K}$.
The fair price for the derivative under $\rho^\ast$ is 
\be
\pi_{\qmat V}^{\rho^\ast} \coloneqq {\rm tr} \left \{ \rho^\ast \frac{\qmat V}{1+r} \right \}.
\ee
The price is given by a trace using the risk-free density operator and the derivative. The price can be maximized or minimized over $\mathcal P$ to obtain the possible interval of prices. In cases where the density operator lives in a large-dimensional Hilbert space, the pricing could be inefficient classically. Pricing such derivatives may require the use of quantum computers. 

The final example considers a case where the price of a quantum asset is lower than the analogous purely diagonal asset, and the risk-neutral density operator has quantum properties. This risk-neutral operator is then used to price a simple derivative which is ``half as quantum''. 
Similar to the previous examples, consider the asset $\qmat S = \frac{a+b}{2} \mathbbm 1 + \frac{a-b}{2}\sigma_z + q  \sigma_x$, where $0 < a < b$ and $q \in[0,\sqrt {ab}]$. 
Let $\pi_0$ be such that $b >  (1+r) \pi_0 >a$ and define $p := (b-(1+r)\pi_0)/(b-a)$.
Let today's price of the asset $\mathcal S$ be $\pi := \pi_0 + 2\eta/(1+r) >0$,
where 
$\eta\in [0, q\sqrt{p(1-p)}]$.
Consider the risk neutral density operator $\rho^\ast := p \ket 0 \bra 0 + (1- p) \ket 1 \bra 1 + \frac{\eta}{q} \sigma_x$,
which is a positive semi-definite matrix, under which $\mathbbm E^{\rho^\ast}\left [\frac{\qmat S}{1+r}\right] = \pi$.
Now let a derivative be 
$\qmat V := \frac{a+b}{2} \mathbbm 1 + \frac{a-b}{2}\sigma_z + \frac{q}{2} \sigma_x$. A fair price is $\mathbbm E^{\rho^\ast}\left [\frac{\qmat V}{1+r}\right] = \pi_0+ \frac{\eta}{(1+r)}$.

\section{Discussion}

We have investigated a direction for quantum finance which directly takes into account the potential influence of quantum technologies on financial markets. This direction is different than the direction of using quantum computers of solving classically-defined financial problems, but eventually, these directions could be combined. There have been related investigations before. One of them notices the connection between the Black-Scholes partial differential equation and the imaginary time Schr\"odinger equation. For the BSM model an effective Hamiltonian $H_{\rm BSM}$ can be defined which leads to the same time evolution \cite{haven2002discussion,
baaquie2003quantum}.
Moreover, there is work on generalizing the Wiener process to a pseudo-Wiener process and quantum stochastic calculus \cite{segal1998black,accardi2007quantum,melnyk2008quantum,bhatnagar2022quantum}. This process shows ballistic dynamics (energy conserving) in contrast to the diffusive dynamics of the Wiener process. A pricing formula for European options can be derived.
In a broader sense, quantum effects have been discussed in economics and finance
\cite{orrell2021quantum,orrell2020quantum,orrell2020quantumBook}.
We also note the topic of quantum money \cite{Molina2013}, where the no-cloning property of quantum states is used to define a currency that cannot be copied.
Hence, the topic of quantizing finance has a reasonably long history. 
Our work shares some similarities in spirit with quantum game theory \cite{eisert1999quantum,khan2018quantum}. The strategies the players can play can be quantum states which can lead to different outcomes in prisoner's dilemma games. We note critiques of quantum game theory \cite{khan2018quantum}. 

We stress that our scenario is technologically motivated. This motivation is in contrast to a motivation arising from the microscopic description of the universe or from molecular processes in the brain, which have been discussed before.
With the advent of quantum technologies (both computation and communication) in financial markets, it is not entirely unreasonable that such a situation arises where the market has to be described with a density operator. It could be argued that in the presence of these quantum technologies, the fully-classical scenario is an unstable point that could readily evolve into a point where quantum effects have to be taken into account.

The present setting allows for the introduction of \textit{quantum derivatives} (Section \ref{secFT}). In our setting, such derivatives are quantum assets (Definition \ref{def:quantum_asset}) for which there is no today's price $\pi$. The pricing problem involves finding a fair price from the existing quantum/classical assets using the martingale density operator discussed in the present work. We leave further investigations of quantum derivatives for future work.
We point out other future directions.
The quantum assets here were defined as having a current price and a future price. This definition does not argue about the intrinsic value of the assets which is left for future work.
In this work, we take $2 \times 2$ density operators for the examples to show the simplest generalization of the classical two-outcome probability vector. As the dimension increases, more interesting quantum financial phenomena could be discovered.
In addition, other basic concepts like the redundancy of the market and market completeness are worthwhile directions.
Finally, one may investigate other forms of the Radon-Nikodym quantum measure change with a stronger foundation in superoperator theory.

\section{Acknowledgements}
We acknowledge valuable discussions with Serge Massar, Sergi Ramos-Calderer, Sai Vinjanampathy, Xiufan Li, and Yanglin Hu.
This research is supported by the National Research Foundation, Singapore, and A*STAR under its CQT Bridging Grant and its Quantum Engineering Programme under grant NRF2021-QEP2-02-P05.

\bibliographystyle{unsrt}
\bibliography{references}

\begin{thebibliography}{10}

\bibitem{Black1973}
Fischer Black and Myron Scholes.
\newblock The pricing of options and corporate liabilities.
\newblock {\em Journal of Political Economy}, 81:637--654, 1973.

\bibitem{Merton1973}
Robert~C. Merton.
\newblock Theory of rational option pricing.
\newblock {\em The Bell Journal of Economics and Management Science},
  4(1):141--183, 1973.

\bibitem{MR0103768}
Harry~M. Markowitz.
\newblock {\em Portfolio selection: {E}fficient diversification of
  investments}.
\newblock Cowles Foundation for Research in Economics at Yale University,
  Monograph 16. John Wiley \& Sons, Inc., New York; Chapman \& Hall, Ltd.,
  London, 1959.

\bibitem{MR3235228}
Harry Markowitz.
\newblock Portfolio selection [reprint of {J}. {F}inance {\bf 7} (1952), no. 1,
  77--91].
\newblock In {\em Financial risk measurement and management}, volume 267 of
  {\em Internat. Lib. Crit. Writ. Econ.}, pages 197--211. Edward Elgar,
  Cheltenham, 2012.

\bibitem{bachelier1900theorie}
Louis Bachelier.
\newblock Th{\'e}orie de la sp{\'e}culation.
\newblock In {\em Annales scientifiques de l'{\'E}cole normale sup{\'e}rieure},
  volume~17, pages 21--86, 1900.

\bibitem{Davis2006Book}
Mark~H. Davis and Alison Etheridge.
\newblock {\em Speculation: Louis Bachelier and the Origins of Modern Finance}.
\newblock Princeton University Press, 2006.

\bibitem{einstein1905molekularkinetischen}
Albert Einstein.
\newblock {\"U}ber die von der molekularkinetischen theorie der w{\"a}rme
  geforderte bewegung von in ruhenden fl{\"u}ssigkeiten suspendierten teilchen.
\newblock {\em Annalen der physik}, 4, 1905.

\bibitem{wiener1923differential}
Norbert Wiener.
\newblock Differential-space.
\newblock {\em Journal of Mathematics and Physics}, 2(1-4):131--174, 1923.

\bibitem{samuelson1973mathematics}
Paul~A Samuelson.
\newblock Mathematics of speculative price.
\newblock {\em Siam Review}, 15(1):1--42, 1973.

\bibitem{samuelson2009enjoyable}
Paul~A Samuelson.
\newblock An enjoyable life puzzling over modern finance theory.
\newblock {\em Annu. Rev. Financ. Econ.}, 1(1):19--35, 2009.

\bibitem{samuelson2015rational}
Paul~A Samuelson.
\newblock Rational theory of warrant pricing.
\newblock In {\em Henry P. McKean Jr. Selecta}, pages 195--232. Springer, 2015.

\bibitem{follmer2004stochastic}
Hans F{\"o}llmer and Alexander Schied.
\newblock {\em Stochastic finance. An introduction in discrete time}.
\newblock Springer New York, 2016.

\bibitem{orus2019quantum}
Roman Orus, Samuel Mugel, and Enrique Lizaso.
\newblock Quantum computing for finance: overview and prospects.
\newblock {\em Reviews in Physics}, 4:100028, 2019.

\bibitem{bouland2020prospects}
Adam Bouland, Wim van Dam, Hamed Joorati, Iordanis Kerenidis, and Anupam
  Prakash.
\newblock Prospects and challenges of quantum finance.
\newblock {\em arXiv preprint arXiv:2011.06492}, 2020.

\bibitem{egger2020quantum}
Daniel~J. Egger, Claudio Gambella, Jakub Marecek, Scott McFaddin, Martin
  Mevissen, Rudy Raymond, Andrea Simonetto, Stefan Woerner, and Elena Yndurain.
\newblock Quantum computing for finance: state of the art and future prospects.
\newblock {\em IEEE Transactions on Quantum Engineering}, 2020.

\bibitem{herman2022survey}
Dylan Herman, Cody Googin, Xiaoyuan Liu, Alexey Galda, Ilya Safro, Yue Sun,
  Marco Pistoia, and Yuri Alexeev.
\newblock A survey of quantum computing for finance.
\newblock {\em arXiv preprint arXiv:2201.02773}, 2022.

\bibitem{barkoutsos2020improving}
Panagiotis~Kl. Barkoutsos, Giacomo Nannicini, Anton Robert, Ivano Tavernelli,
  and Stefan Woerner.
\newblock Improving variational quantum optimization using {CV}a{R}.
\newblock {\em Quantum}, 4:256, 2020.

\bibitem{dasgupta2019quantum}
Samudra Dasgupta and Arnab Banerjee.
\newblock Quantum annealing algorithm for expected shortfall based dynamic
  asset allocation.
\newblock {\em arXiv preprint arXiv:1909.12904}, 2019.

\bibitem{hodson2019portfolio}
Mark Hodson, Brendan Ruck, Hugh Ong, David Garvin, and Stefan Dulman.
\newblock Portfolio rebalancing experiments using the quantum alternating
  operator ansatz.
\newblock {\em arXiv preprint arXiv:1911.05296}, 2019.

\bibitem{rebentrost2018quantum}
Patrick Rebentrost and Seth Lloyd.
\newblock Quantum computational finance: quantum algorithm for portfolio
  optimization.
\newblock {\em arXiv preprint arXiv:1811.03975}, 2018.

\bibitem{alcazar2020classical}
Javier Alcazar, Vicente Leyton-Ortega, and Alejandro Perdomo-Ortiz.
\newblock Classical versus quantum models in machine learning: insights from a
  finance application.
\newblock {\em Machine Learning: Science and Technology}, 1(3):035003, 2020.

\bibitem{Woerner2018}
Stefan Woerner and Daniel~J. Egger.
\newblock Quantum risk analysis.
\newblock {\em arXiv preprint arXiv:1806.06893}, 2018.

\bibitem{alcazar2022quantum}
Javier Alcazar, Andrea Cadarso, Amara Katabarwa, Marta Mauri, Borja Peropadre,
  Guoming Wang, and Yudong Cao.
\newblock Quantum algorithm for credit valuation adjustments.
\newblock {\em New Journal of Physics}, 24(2):023036, 2022.

\bibitem{han2022quantum}
Jeong~Yu Han and Patrick Rebentrost.
\newblock Quantum advantage for multi-option portfolio pricing and valuation
  adjustments.
\newblock {\em arXiv preprint arXiv:2203.04924}, 2022.

\bibitem{Rebentrost2018finance}
Patrick Rebentrost, Brajesh Gupt, and Thomas~R. Bromley.
\newblock Quantum computational finance: Monte carlo pricing of financial
  derivatives.
\newblock {\em Phys. Rev. A}, 98:022321, 2018.

\bibitem{Martin2019}
Ana Martin, Bruno Candelas, {\'A}ngel Rodr{\'i}guez-Rozas, Jos{\'e}~D.
  Mart{\'i}n-Guerrero, Xi~Chen, Lucas Lamata, Rom{\'a}n Or{\'u}s, Enrique
  Solano, and Mikel Sanz.
\newblock Toward pricing financial derivatives with an ibm quantum computer.
\newblock {\em Physical Review Research}, 3(1):013167, 2021.

\bibitem{vazquez2021efficient}
Almudena~Carrera Vazquez and Stefan Woerner.
\newblock Efficient state preparation for quantum amplitude estimation.
\newblock {\em Physical Review Applied}, 15(3):034027, 2021.

\bibitem{stamatopoulos2021towards}
Nikitas Stamatopoulos, Guglielmo Mazzola, Stefan Woerner, and William~J. Zeng.
\newblock Towards quantum advantage in financial market risk using quantum
  gradient algorithms.
\newblock {\em arXiv preprint arXiv:2111.12509}, 2021.

\bibitem{an2021quantum}
Dong An, Noah Linden, Jin-Peng Liu, Ashley Montanaro, Changpeng Shao, and Jiasu
  Wang.
\newblock Quantum-accelerated multilevel monte carlo methods for stochastic
  differential equations in mathematical finance.
\newblock {\em Quantum}, 5:481, 2021.

\bibitem{enwiki:1140135669}
{Wikipedia contributors}.
\newblock Quantum machine learning --- {Wikipedia}{,} the free encyclopedia.
\newblock
  \url{https://en.wikipedia.org/w/index.php?title=Quantum_machine_learning&oldid=1140135669},
  2023.
\newblock [Online; accessed 31-March-2023].

\bibitem{aimeur2006machine}
Esma A{\"\i}meur, Gilles Brassard, and S{\'e}bastien Gambs.
\newblock Machine learning in a quantum world.
\newblock In {\em Advances in Artificial Intelligence: 19th Conference of the
  Canadian Society for Computational Studies of Intelligence, Canadian AI 2006,
  Qu{\'e}bec City, Qu{\'e}bec, Canada, June 7-9, 2006. Proceedings 19}, pages
  431--442. Springer, 2006.

\bibitem{dunjko2016quantum}
Vedran Dunjko, Jacob~M Taylor, and Hans~J Briegel.
\newblock Quantum-enhanced machine learning.
\newblock {\em Physical review letters}, 117(13):130501, 2016.

\bibitem{radcliffe1990investment}
R.C. Radcliffe.
\newblock {\em Investment: concepts, analysis, strategy}.
\newblock HarperCollins College Publisher, 1990.

\bibitem{haven2002discussion}
Emmanuel~E Haven.
\newblock A discussion on embedding the black--scholes option pricing model in
  a quantum physics setting.
\newblock {\em Physica A: Statistical Mechanics and its Applications},
  304(3-4):507--524, 2002.

\bibitem{baaquie2003quantum}
Belal~E Baaquie, Claudio Coriano, and Marakani Srikant.
\newblock Quantum mechanics, path integrals and option pricing: Reducing the
  complexity of finance.
\newblock In {\em Nonlinear Physics: Theory and Experiment II}, pages 333--339.
  World Scientific, 2003.

\bibitem{segal1998black}
Wiliam Segal and IE~Segal.
\newblock The black--scholes pricing formula in the quantum context.
\newblock {\em Proceedings of the National Academy of Sciences},
  95(7):4072--4075, 1998.

\bibitem{accardi2007quantum}
Luigi Accardi and Andreas Boukas.
\newblock The quantum black-scholes equation.
\newblock {\em arXiv preprint arXiv:0706.1300}, 2007.

\bibitem{melnyk2008quantum}
SI~Melnyk and IG~Tuluzov.
\newblock Quantum analog of the black-scholes formula (market of financial
  derivatives as a continuous weak measurement).
\newblock {\em Electronic Journal of Theoretical Physics (EJTP)}, 5(18), 2008.

\bibitem{bhatnagar2022quantum}
Anantya Bhatnagar and Dimitri~D Vvedensky.
\newblock Quantum effects in an expanded black-scholes model.
\newblock {\em arXiv preprint arXiv:2203.07940}, 2022.

\bibitem{orrell2021quantum}
David Orrell.
\newblock A quantum walk model of financial options.
\newblock {\em Wilmott}, 2021(112):62--69, 2021.

\bibitem{orrell2020quantum}
David Orrell.
\newblock A quantum model of supply and demand.
\newblock {\em Physica A: statistical Mechanics and its Applications},
  539:122928, 2020.

\bibitem{orrell2020quantumBook}
David Orrell.
\newblock {\em Quantum Economics and Finance: An Applied Mathematics
  Introduction}.
\newblock Panda Ohana Publishing, 2020.

\bibitem{Molina2013}
Abel Molina, Thomas Vidick, and John Watrous.
\newblock Optimal counterfeiting attacks and generalizations for wiesner's
  quantum money.
\newblock In Kazuo Iwama, Yasuhito Kawano, and Mio Murao, editors, {\em Theory
  of Quantum Computation, Communication, and Cryptography}, pages 45--64,
  Berlin, Heidelberg, 2013. Springer Berlin Heidelberg.

\bibitem{eisert1999quantum}
Jens Eisert, Martin Wilkens, and Maciej Lewenstein.
\newblock Quantum games and quantum strategies.
\newblock {\em Physical Review Letters}, 83(15):3077, 1999.

\bibitem{khan2018quantum}
Faisal~Shah Khan, Neal Solmeyer, Radhakrishnan Balu, and Travis~S Humble.
\newblock Quantum games: a review of the history, current state, and
  interpretation.
\newblock {\em Quantum Information Processing}, 17(11):1--42, 2018.

\bibitem{Belavkin1986}
Viacheslav~P. Belavkin and Przemyslaw Staszewski.
\newblock A radon-nikodym theorem for completely positive maps.
\newblock {\em Reports on Mathematical Physics}, 24:49--55, 1986.

\bibitem{Raginsky2003}
Maxim Raginsky.
\newblock Radon–nikodym derivatives of quantum operations.
\newblock {\em Journal of Mathematical Physics}, 44(11):5003--5020, 2003.

\end{thebibliography}

\newpage

\appendix
\section{Preliminaries} \label{secPrelim}

\subsection{Mathematical notations}
We denote by $[n]$ the set $\{1,\dots,n\}$ and by $[n]_0$ the set $\{0,1,\dots,n\}$.
We denote vectors with bold letters, e.g., $\qvec s$, matrices by calligraphic letters, e.g., $\qmat S$, and tensors by calligraphic bold letters, e.g., $\qten S$. Given two vectors $\qvec u, \qvec v$ we denote their inner product by $\qvec u \cdot \qvec v$. More general, given a vector $\qvec u \in \mathbbm R^d$ and an array of matrices (tensor) $ \qten S=(\qmat S_1,\qmat S_2,...,\qmat S_d)$, where $S_j \in \mathbbm R^{K \times K}$ for all $j \in [d]$,
we denote the contraction along the first dimension also as $\qvec u \cdot \qten S = \sum_{j=1}^d u_j \qmat S_j$. 

An important result used in classical arbitrage theory is the separating hyperplane theorem. 
\begin{theorem}[Separating hyperplane theorem \cite{follmer2004stochastic}]\label{thm:seperating_hyperplane}
    Suppose that $\mathcal C \subset \mathbbm R^n$ is a nonempty convex set with $\qvec 0 \notin \mathcal C$. Then there exists $\qvec \eta \in \mathbbm R^n$ with $\qvec \eta \cdot \qvec x \geq 0$ for all $\qvec x \in \mathcal C$, and with $\qvec \eta \cdot \qvec x^\ast$ for at least $\qvec x^\ast > 0$ for at least one $\qvec x^\ast \in \mathcal C$.
\end{theorem}

\subsection{Stochastic finance}
Here we give a brief introduction to a classical discrete model of mathematical finance \cite{follmer2004stochastic}. We use a one-period model and leave the generalization to multi-period models for future work. At time $t=0$, we have a vector of prices $\overline {\qvec \pi} = (\pi_0,\qvec \pi) = (\pi_0,\dots,\pi_d) \in \mathbbm R^{d+1}_{+}$ of one risk-free asset and $d$ risky assets. We can choose a portfolio $\overline {\qvec \xi} = (\xi_0,\qvec \xi)=(\xi_0,\dots,\xi_d) \in \mathbbm R^{d+1}$. The value of the portfolio is $ \overline {\qvec \pi} \cdot \overline {\qvec \xi}$ at time $t=0$. At time $t=1$, we receive the price vector of the assets $\overline {\qvec S}=(S_0,\qvec S)=(S_0,S_1,\dots,S_d) \in \mathbbm R^{d+1}_{+}$. From the point of view at time $t=0$ each element is a random variable except $S^0$ which is equal to $\pi_0 (1+r)$, where $r$ is the interest rate. For simplicity, the sample space $\Omega$ in our model is finite, i.e. $|\Omega|=K$ and we assume that each element of $\Omega$ happens with strictly positive probability. 
Therefore for the probability measure $\mathbbm P$, we can use $\qvec p$, a probability vector where $p_i = \mathbbm P(w_i)$ for all $w_i \in \Omega$ and $i \in [K]$. The expected value of our portfolio at time $t=1$ is therefore
$$
\mathbbm E^{\mathbbm P}[\overline {\qvec \xi} \cdot \overline {\qvec S}]
= \xi_0 s_0 + \mathbbm E^{\mathbbm P}[\qvec \xi \cdot \qvec S]
= \xi_0\pi_0(1+r)+\sum_{\omega \in \Omega}\mathbbm P(\omega)\sum_{j=1}^n \xi_j S_j(\omega).
$$
A \textit{risk-neutral measure} $\mathbbm P^\ast$ is a probability measure such that
$$
\pi_i = \mathbbm E^{\mathbbm P^\ast}\left [\frac{S_i}{1+r}\right],
$$
for all $i \in [N]_0$. In words, under $\mathbbm P^\ast$, the discounted expected value of the future asset prices corresponds exactly to today's prices.
We recall the definition of arbitrage from \cite{follmer2004stochastic} here.
\begin{definition}[Arbitrage Opportunity \cite{follmer2004stochastic}]
\label{def:arbitrage_opportunity}
A portfolio $\overline {\qvec \xi} \in \RR^{d+1}$ is called an \textit{arbitrage opportunity} if 
\begin{itemize}
\item $\overline {\qvec \xi} \cdot \overline {\qvec \pi} \leq 0$, and
\item for all $\omega \in \Omega$, $\overline {\qvec \xi} \cdot \overline {\qvec S}(\omega) \geq 0$, and
\item there exists at least one $\omega \in \Omega$ such that $\overline {\qvec \xi} \cdot \overline {\qvec S}(\omega) > 0$.
\end{itemize} 
\end{definition}

\section{Arbitrage theory} \label{appArb}

\textbf{Explanation of Figure \ref{figArb}.} Considering Def.~\ref{def:x_arbitrage_opportunity}, Figure \ref{figArb} illustrates the definition of quantum arbitrage, in particular the comparison $\hat{\mathcal H}=\Omega$ and $\hat{\mathcal H}=\mathcal H^\Omega$. All other assumptions are kept the same, i.e., let $(\overline {\qvec \pi}, \overline{\qten S})$ be a price system and $\rho \in \mathcal D(\mathcal H^\Omega)$ be the market density operator. 
Confer Def.~\ref{def:arbitrage_opportunity} for the strictly classical definition. 
Define the base set
\be
B :=  \left \{ \overline {\qvec \xi} \in \RR^{d+1}\ \vert\ \overline {\qvec \xi} \cdot \overline {\qvec \pi} \leq 0 \right \},
\ee
and following sets
\be
C1 &:=& \left \{ \overline {\qvec \xi} \in B\ \vert\ \forall \ket \omega \in  \Omega \text{ with }\bra \omega \rho \ket \omega > 0, \text{ it holds that }\bra \omega \overline {\qvec \xi} \cdot \overline {\qten S} \ket \omega \geq 0 \right \},  \\
C2 &:=& \left \{\overline {\qvec \xi} \in B \ \vert\
\exists \ket \omega \in \Omega \text{ with }\bra \omega \rho \ket \omega > 0\text{ and }\bra \omega \overline {\qvec \xi} \cdot \overline {\qten S}) \ket \omega > 0 \right \}. 
\ee
In addition, from Def.~\ref{def:x_arbitrage_opportunity}, define
\be
Q1 &:=& \left \{ \overline {\qvec \xi}\in B \ \vert\ \forall \ket \psi \in  \hat{\mathcal H} \text{ with }\bra \psi \rho \ket \psi > 0, \text{ it holds that }\bra \psi \overline {\qvec \xi} \cdot \overline {\qten S} \ket \psi \geq 0 \right \},  \\
Q2 &:=& \left \{\overline {\qvec \xi} \in B \ \vert\
\exists \ket \psi \in \hat{\mathcal H} \text{ with }\bra \psi \rho \ket \psi > 0\text{ and }\bra \psi \overline {\qvec \xi} \cdot \overline {\qten S}) \ket \psi > 0 \right \}. 
\ee
The following simple relationships can be deduced
\be
Q1 &\subseteq& C1 \\
C2 &\subseteq& Q2.
\ee
We write the definition of classical and quantum arbitrage as 
\be 
CA &:=&C1 \cap C2 \\
QA &:=&Q1 \cap Q2.
\ee
These sets and their relationships are drawn in Figure \ref{figArb}.
\begin{lemma}[Arbitrage Lemma]\label{lem:x_arbitrage_risky_asset}
The market model admits a quantum arbitrage opportunity if and only if there is a vector $\qvec \xi \in \mathbbm R^d$ such that
\begin{enumerate}
\item for all $\ket \psi \in \hat{ \mathcal H}$ with $\bra \psi \rho \ket \psi > 0$, we have that $ \bra \psi \qvec \xi \cdot \qten S \ket \psi \geq (1+r) \qvec \xi \cdot \qvec \pi$ and 
\item there exists at least one $\ket \psi \in \hat{ \mathcal H}$ with $\bra \psi \rho \ket \psi > 0$, such that $\bra \psi \qvec \xi \cdot \qten S \ket \psi > (1+r) \qvec \xi \cdot \qvec \pi.$
\end{enumerate}
\end{lemma}

\begin{proof}
$("\Longrightarrow")$. Let $\overline {\qvec \xi}$ be the arbitrage opportunity hence $0 \geq \overline{\qvec \xi} \cdot \overline {\qvec \pi}=\xi_0 \pi_0+\qvec \xi \cdot \qvec \pi$. For all $\ket \psi \in \hat{ \mathcal H}$, we have that $\bra \psi \qvec \xi \cdot \qten S \ket \psi - (1+r) \qvec \xi \cdot \qvec \pi \geq  \bra \psi \qvec \xi \cdot \qten S \ket \psi + (1+r)\xi_0\pi_0 = \bra \psi \qvec \xi \cdot \qten S \ket \psi + (1+r)\bra \psi \xi_0 \pi_0 \ket \psi = \bra \psi \overline {\qvec \xi} \cdot \overline {\qten S} \ket \psi$. From Def. \ref{def:x_arbitrage_opportunity}, we have that $\bra \psi \overline {\qvec \xi} \cdot \overline {\qten S} \ket \psi \geq 0$ for all $\ket \psi \in \hat{ \mathcal H}$ with  $\bra \psi \rho \ket \psi > 0$, and $\bra \psi \overline {\qvec \xi} \cdot \overline {\qten S} \ket \psi>0$ for at least one $\ket \psi \in \hat{ \mathcal H}$ with $\bra \psi \rho \ket \psi > 0$.  Moving the $(1+r)\qvec \xi \cdot \qvec \pi$ to the right-hand side, we obtain the two conditions stated in the lemma.\\
$("\Longleftarrow")$. Let $\qvec \xi$ be the portfolio satisfying the two conditions stated in the lemma. We construct a portfolio $(\xi_0, \qvec \xi)$ with $\xi_0=-\qvec \xi \cdot \qvec \pi / \pi_0$. Then $\overline {\qvec \xi} \cdot \overline {\qvec \pi}=\xi_0\pi_0+\qvec \xi \cdot \qvec \pi = 0$. Moreover, for any $\ket \psi \in \hat{ \mathcal H}$, we have that $\bra \psi \overline {\qvec \xi} \cdot \overline {\qten S} \ket \psi = -(1+r)\qvec \xi \cdot \qvec \pi + \bra \psi \qvec \xi \cdot \qten S \ket \psi$, which is non-negative for all $\ket \psi \in \hat{ \mathcal H}$ with $\bra \psi \rho \ket \psi > 0$ and strictly positive for some $\ket \psi \in \hat{ \mathcal H}$ with $\bra \psi \rho \ket \psi > 0$. 
By Def. \ref{def:x_arbitrage_opportunity}, the portfolio $(\xi_0, \qvec \xi)$ is an arbitrage opportunity.
\end{proof}

\begin{corollary}[Absence of arbitrage opportunity]\label{cor:x_arbitrage_discounted_net_gains}
The market is absent of arbitrary opportunity if and only if for all $\qvec \xi \in \mathbbm R^d$, the fact that $\bra \psi \qvec \xi \cdot \qten Y \ket \psi \geq 0 $ for all $\ket \psi \in \hat{ \mathcal H}$ with $\bra \psi \rho \ket \psi > 0$ implies that $\bra \psi \qvec \xi \cdot \qten Y \ket \psi = 0 $ for all $\ket \psi \in \mathcal H^\Omega$ with $\bra \psi \rho \ket \psi > 0$.
\end{corollary}  
\begin{proof}
From Lemma \ref{lem:x_arbitrage_risky_asset}, no arbitrage means that there exists no $\ket \psi \in \hat{ \mathcal H}$ with $\bra \psi \rho \ket \psi > 0$, such that the statement of Point 2. holds. Hence the net gains have to be zero. 
\end{proof}

\section{Quantum measure change} \label{secRN}

Our stochastic model for the financial market is described by a density matrix encoding the probabilities for the states in $\mathcal H^\Omega$. 
An important notion in financial theory is the change of measure from one probability measure to another. In this section, we define the necessary mathematical terminology for density operators and show a measure change (Radon-Nikodym derivative) that performs correctly in the context of this work.
First, we would like to mimic the notion of absolute continuity in measure theory, and we define it here for density operators.

\begin{definition}[Absolute continuity]\label{def:absolutely-continuity}
Suppose $\rho$ and $\sigma$ are two density operators on $\mathcal H^\Omega$. Density operator $\sigma$ is said to be absolutely continuous with respect to $\rho$ on $\mathcal H^\Omega$, and we write $\sigma \ll \rho$, if for all $\ket \psi \in \mathcal H^\Omega$
$$
\rho \ket \psi = \qvec 0 \implies \sigma \ket \psi = \qvec 0.
$$
In other words, for the nullspaces, it holds that ${\rm null}(\rho) \subseteq {\rm null}(\sigma)$.
\end{definition}
We note that for all $\ket \psi \in \mathcal H^\Omega$, we have that $\rho \ket \psi = \qvec 0 \iff \bra \psi \rho \ket \psi =0$ for positive semi-definite operators. For this statement, the  direction $\impliedby$ is proved by the standard property $\vert \bra {\psi_\perp} \rho \ket \psi \vert \leq \sqrt{\bra {\psi_\perp} \rho \ket {\psi_\perp} \bra {\psi} \rho \ket \psi}$ for orthogonal states $\ket \psi$ and $\ket {\psi_\perp}$.
Next, we specialize the definition of \textit{equivalence} of probability measures to the same notion in the context of density operators.
\begin{definition}[Null-space equivalence between density operators]\label{def:equivalence-between-density-matrices}
Given two density operators on $\mathcal H^\Omega$, if both $\sigma \ll \rho$ and $\rho \ll \sigma$ hold, we say that $\rho$ and $\sigma$ are \textit{equivalent}, and we write $\rho \approx \sigma$.
\end{definition}

As it will appear in the main theorem of the present work, we define the \textit{quantum Radon-Nikodym derivative} to describe the transformation between density operators.

\begin{definition}[Density operator Radon-Nikodym derivative]\label{def:density_operator_derivative}
Given two density operators $\rho, \sigma \in \mathbbm C^{K \times K}$ with $\sigma \ll \rho$, a quantum map $\varphi : \mathbbm C^{K \times K} \to  \mathbbm C^{K \times K}$ is called \textit{density operator Radon-Nikodym derivative} from $\rho$ to $\sigma$, if $\varphi [\cdot] = \sigma^{1/2} \rho^{-1/2} [\cdot] \rho^{-1/2} \sigma^{1/2}$, where $\rho^{-1}$ is the pseudo-inverse of $\rho$ and $\sigma^{1/2}$ applies the square-root to the eigenvalues of $\sigma$.
We also denote $\varphi(\sigma, \rho)[\cdot]$ to clarify the parameters of the superoperator. 
\end{definition}
Note that $\varphi(\sigma, \rho)[\rho] = \sigma^{1/2} \rho^{-1/2} \rho \rho^{-1/2} \sigma^{1/2} = \sigma^{1/2}\Pi_{+}(\rho) \sigma^{1/2} = \sigma^{1/2}(\mathbbm 1 - \Pi_{\rm null}(\rho)) \sigma^{1/2} = \sigma$.
Here, $\Pi_{+}(\rho)$ is the projector into the row-space of $\rho$ and $\Pi_{\rm null}(\rho)$ is the projector into the null-space of $\rho$. We have used that $\sigma^{1/2} \Pi_{\rm null}(\rho) \sigma^{1/2}  = 0$ because of $\sigma \ll \rho$.
\begin{lemma}            
The density operator Radon-Nikodym derivative satisfies some usual properties of a measure change.           
\end{lemma}
\begin{proof}
Let $\rho$, $\sigma$, and $\tau$ be density operators.  
i) If $\rho \ll \tau$ and $\sigma \ll \tau$, then we have linearity in the first parameter given that the argument is the same as the second parameter, by showing that
\be
\varphi(\rho + \sigma, \tau)[\tau] &=& 
(\rho + \sigma)^{1/2} \tau^{-1/2} \tau \tau^{-1/2} (\rho + \sigma)^{1/2}\\
&=&\rho + \sigma \\ 
&=& \rho^{1/2} \tau^{-1/2} \tau \tau^{-1/2} \rho^{1/2} + \sigma^{1/2} \tau^{-1/2} \tau \tau^{-1/2} \sigma^{1/2}\\
&=&
\varphi(\rho, \tau)[\tau] + \varphi(\sigma, \tau)[\tau].
\ee
ii) Let $\rho \ll \sigma \ll \tau$, then
\be
\varphi(\rho, \tau)[\cdot] &=&\rho^{1/2} \tau^{-1/2} [\cdot] \tau^{-1/2} \rho^{1/2} \\ 
&=&\rho^{1/2} \sigma^{-1/2} \sigma^{1/2} \tau^{-1/2} [\cdot] \tau^{-1/2} \sigma^{-1/2} \sigma^{1/2} \rho^{1/2} \\
&=&
\varphi(\rho, \sigma)[\varphi(\sigma, \tau)[\cdot]].
\ee
iii) Let $\rho \approx \sigma$, then 
\be
\varphi(\rho, \sigma)[\varphi(\sigma, \rho)[\cdot]] &=& \rho^{1/2}\sigma^{-1/2} \sigma^{1/2} \rho^{-1/2} [\cdot] \rho^{-1/2} \sigma^{1/2}\sigma^{-1/2}\rho^{1/2}\\
&=& \Pi_{+}(\rho) [\cdot] \Pi_{+}(\rho).
\ee
The expression evaluates to $\varphi(\rho, \sigma)[\varphi(\sigma, \rho)[\tau]] = \tau$ for $\tau \ll \rho$
and motivates the definition of the inverse $\varphi^{-1}(\sigma, \rho) := \varphi(\rho, \sigma)$ of $\varphi(\sigma,\rho)$.
\\
iv) Let $X \in \mathbbm C^{K\times K}$ and $\sigma \ll \rho$. Then
\be
{\rm tr}\left \{ \sigma X \right\} = {\rm tr}\left \{ \varphi(\sigma,\rho)[\rho] X \right\} = {\rm tr}\left \{ \rho Y\right \},
\ee
with $Y:= \rho^{-1/2} \sigma^{1/2} X \sigma^{1/2}\rho^{-1/2} \equiv \varphi (\rho^{-1},\sigma^{-1})[X]$.
\end{proof}
In the context of the present work, the definition of the Radon-Nikodym derivative superoperator operates as expected.
We note Refs.~\cite{Belavkin1986,Raginsky2003} for a mathematical treatment of the Radon-Nikodym derivative in the non-commutative setting.

\section{Proof of both directions of Theorem \ref{thmMain}} \label{appProofMain}

\begin{proof}[Proof of the direction
$("\Longleftarrow")$ of Theorem \ref{thmMain}]
This direction shows that the existence of a risk-neutral density operator implies the absence of arbitrage.
Suppose there exists a risk-neutral density operator $\rho^\ast \in \mathcal P$. With the original density operator $\rho$, take a portfolio $\overline {\qvec \xi} \in \RR^{d+1}$ which satisfies the two conditions from the definition of quantum arbitrage Def.~\ref{def:x_arbitrage_opportunity}: 
(1) For all $\ket \psi \in \mathcal H^\Omega$ with $\bra \psi \rho \ket \psi>0$, we have that $\bra \psi \overline {\qvec \xi} \cdot \overline {\qten S} \ket \psi \geq 0$ and (2) ${\rm tr} \left\{ \rho\ \overline{\qvec \xi} \cdot \overline {\qten S} \right \}>0$. 
Condition (1) is true also for density operator $\rho^\ast$ because of the set identity $\left\{ \ket \psi \in \mathcal H^\Omega : \bra \psi \rho \ket \psi>0 \right\} =\left\{ \ket \psi \in \mathcal H^\Omega : \bra \psi \rho^\ast \ket \psi>0\right \}$ for $\rho \approx \rho^\ast$. 
From condition (2), there exist some $\ket\psi \in \mathcal H^\Omega $ for which $\bra \psi \rho \ket \psi>0$ and $\bra \psi \overline{\qvec \xi} \cdot \overline {\qten S}\ket \psi >0$.
For the same $\ket\psi$ we have that $\bra \psi \rho^\ast \ket \psi>0$ because of $\rho \approx \rho^\ast$,
Hence, condition (2) is true also for density operator $\rho^\ast$, because of 
${\rm tr}\left\{ \rho^\ast \overline {\qvec \xi} \cdot \overline {\qten S}\right\} > \bra \psi \rho^\ast \ket \psi \bra {\psi}\overline {\qvec \xi} \cdot \overline {\qten S} \ket{\psi}>0$. Hence,
\be
\overline {\qvec \xi} \cdot \overline {\qvec \pi} = \sum_{i=0}^d\xi_i \cdot \pi_i = \sum_{i=0}^{d}{\rm tr} \left\{ \rho^\ast \frac{\xi_i \cdot \mathcal S_i}{1+r}\right\}= {\rm tr} \left\{ \rho^\ast \frac{\overline {\qvec \xi} \cdot \overline {\qten S}}{1+r} \right\} > 0,
\ee
where for the second equation we use the definition in Eq.~(\ref{eqMartingaleDensityOperator}). Thus $\overline {\qvec \xi}$ cannot be an arbitrage opportunity.
\end{proof}

\begin{proof}[Proof of the direction
$("\Longrightarrow")$ of Theorem \ref{thmMain}]
Consider again the discounted net gains $\qmat Y_i \coloneqq \frac{\qmat S_i}{1+r}-\pi_i \mathbbm I$, for $i\in[d]_0$. As $\rho$ is a density operator, we have that ${\rm tr}\left\{\rho\vert \qmat Y_i \vert \right \} \leq \frac{1}{1+r}{\rm tr}\left\{\rho \qmat S_i \right\} + \pi_i$, which is finite by assumption. Let $\mathcal Q$ be the set of all density operators which are equivalent to $\rho$, i.e., 
\be
\mathcal Q := \{ \sigma : \sigma\ \text{density operator with}\ \sigma\approx\rho \}.
\ee
We prove that $\mathcal Q$ is convex. Let $\sigma_a := a\sigma_0+(1-a)\sigma_1$, where $\sigma_0,\sigma_1 \in \mathcal Q$ and $0 \leq a \leq 1$.
We have that $\sigma_a$ is a density operator because it is a convex combination of density operators. 
For proving equivalence, we use that for all $\ket \psi \in \mathcal H^\Omega$,
\be
\bra \psi \sigma_a \ket \psi = 0 \iff \bra \psi \sigma_0 \ket \psi =0\ \mathrm {and}\ \bra \psi \sigma_1 \ket \psi = 0.
\ee
So, $\sigma_a \in \mathcal Q$ and the set $\mathcal Q$ is a convex set. 
Let the set of expected net gain vectors be
\be
\mathcal C :=\{{\rm tr}\{\sigma \qten Y \}:\sigma \in \mathcal Q\} \subseteq \mathbbm R^{d}.
\ee
To prove convexity of $\mathcal C$, take $\qvec e_0,\qvec e_1 \in \mathcal C$. Let $\sigma_0$ and $\sigma_1$ be the associated density operators, and $\sigma_a$ be the convex combination as above.
By convexity of $\mathcal Q$, we have that $\sigma_a \in \mathcal Q$ and hence $a  {\rm tr}\{\sigma_0 \qten Y\}+(1-a)  {\rm tr}\{\sigma_1 \qten Y\} = {\rm tr}\{\sigma_a  \qten Y \} \in \mathcal C$. 
Now, we need to prove $\mathcal C$ contains the origin. 
Assume by contradiction, $\qvec 0 \notin \mathcal C$, then by Theorem \ref{thm:seperating_hyperplane}, we have a vector $\qvec \xi \in \RR^d$ such that for all $\qvec x \in \mathcal C$, we have that $\qvec \xi \cdot \qvec x \geq 0$ and for some $\qvec x^\ast \in \mathcal C$, we have that $\qvec \xi \cdot \qvec x^\ast > 0$. 
Hence, this $\qvec \xi$ satisfies 
\be
&\text{(i)}& {\rm tr}\{ \sigma \qvec \xi \cdot \qten Y \} \geq 0\ \text{for all}\ \sigma \in \mathcal Q \\ 
&\text{(ii)}& {\rm tr} \{ \sigma_0 \qvec \xi \cdot \qten Y\} > 0\ \text{for some}\ \sigma_0 \in \mathcal Q. 
\ee
Condition (ii) implies that there exists $\ket \psi \in \mathcal H^\Omega$ with $\bra \psi \sigma_0 \ket \psi >0$ and $\bra \psi \qvec \xi \cdot \qten Y \ket \psi >0 $, see Lemma \ref{lemStriclyGreater}, and for this $\ket \psi$ we also have $\bra \psi \rho \ket \psi >0$ because of $\sigma_0 \approx \rho$.
Regarding condition i), we now prove that it implies that $\bra \psi \qvec \xi \cdot \qten Y \ket \psi \geq 0$ for all $\ket \psi \in \mathcal H^\Omega$ with $\bra \psi \rho \ket \psi >0$. 
This contradicts our assumption of the absence of arbitrage, which by Corollary \ref{cor:x_arbitrage_discounted_net_gains} implies that there must also be strictly negative outcomes.
Let 
\be 
A := \left \{ \ket \omega \in \Omega :  \bra \omega \qvec \xi \cdot \qten Y \ket \omega <0 \right\}.
\ee
and its complement on $\Omega$ be $A^c$. Note that only computational basis states $\ket \omega$ define the set $A$. Define the projector 
\be
\Pi_A = \sum_{\ket \omega \in A} \ket \omega \bra \omega.
\ee
For integers $n\geq 1$, define the quantum map 
\be
F_n[\cdot] := \left(1-\frac{1}{n}\right) \Pi_A [\cdot]  \Pi_A  +\frac{1}{n}\Pi_{A^c} [\cdot]  \Pi_{A^c},
\ee
and its $\rho$-dependent normalization $f_n := {\rm tr} \left\{ F_n [\rho]\right\}$. 
Note that for all $\ket \psi \in \mathcal H^\Omega$ and density operators $\tau$, we have that
\be 
0 < \bra \psi F_n [\tau] \ket \psi \leq \bra \psi F_n [\mathbbm 1] \ket \psi \leq 1.
\ee
Define density operators $\sigma_n$ implicitly via the superoperator
\be
\varphi(\sigma_n,\rho)[\cdot] = \frac{F_n[\cdot]}{f_n}.
\ee
We show that $\sigma_n \approx \rho$ by 
\be
\bra \psi \sigma_n \ket \psi = 0 &\iff& \bra \psi \varphi(\sigma_n,\rho)[\rho] \ket \psi = 0 \\
&\iff& \bra \psi \left(1-\frac{1}{n}\right) \Pi_A \rho  \Pi_A  +\frac{1}{n}\Pi_{A^c} \rho \Pi_{A^c}\ket \psi =0 
\\
&\iff& \bra \psi \Pi_A \rho  \Pi_A\ket \psi =0\ {\rm and}\ \bra \psi \Pi_{A^c} \rho \Pi_{A^c}\ket \psi =0
\\
&\iff&\bra \psi \rho \ket \psi =0.
\ee
Hence, $\sigma_n \in \mathcal Q$. 
From condition (i), we know that
\be \label{eqGainLowerB}
0 \leq \qvec \xi \cdot {\rm tr}\{ \sigma_n \qten Y\} &=& \frac{1}{f_n}  {\rm tr}\{ \varphi(\sigma_n,\rho)[\rho]\ \qvec\xi \cdot \qten Y\}\\
&=& \frac{1}{f_n}  {\rm tr}\left \{ \left[\left(1-\frac{1}{n}\right) \Pi_A \rho  \Pi_A  +\frac{1}{n}\Pi_{A^c} \rho \Pi_{A^c}\right] \qvec\xi \cdot \qten Y \right\}\\
&=& \frac{1}{f_n}  {\rm tr}\left \{ \left[\left(1-\frac{1}{n}\right) \Pi_A \qvec\xi \cdot \qten Y   \Pi_A  +\frac{1}{n}\Pi_{A^c} \qvec\xi \cdot \qten Y  \Pi_{A^c}\right] \rho \right\}.
\ee
Taking limits results in 
\be
\lim_{n \to \infty} f_n = {\rm tr} \left\{ \Pi_A \rho \Pi_A\right\} = {\rm tr} \left\{ \Pi_A \rho\right\},
\ee
and
\be
0 \leq \lim_{n \to \infty} {\rm tr}\{ \varphi(\sigma_n,\rho)[\rho]\ \qvec\xi \cdot \qten Y\} = {\rm tr}\left \{ \Pi_A \qvec\xi \cdot \qten Y \Pi_A \rho\right\}.
\ee
We can show that $\Pi_A \qvec\xi \cdot \qten Y \Pi_A$ is a negative semi-definite operator via
\be
\bra \chi \Pi_A \qvec\xi \cdot \qten Y \Pi_A \ket \chi &=& \bra \chi \sum_{\ket \omega \in A} \ket \omega \bra \omega \qvec\xi \cdot \qten Y \sum_{\ket {\omega'} \in A} \ket {\omega'} \bra {\omega'}  \ket{ \chi}\\
&=& \bra {\chi_A} \qvec\xi \cdot \qten Y   \ket{ \chi_A}\leq 0.
\ee
by the definition of $A$.
Hence it follows from Eq.~\ref{eqGainLowerB} that
\be
{\rm tr}\left \{ \Pi_A \qvec\xi \cdot \qten Y \Pi_A \rho\right\} =0.
\ee
Hence, $\qvec\xi \cdot \qten Y$ is a positive semi-definite matrix and thus
for all $\ket \psi \in \mathcal H^\Omega$, $\bra \psi \qvec \xi \cdot \qten Y \ket \psi \geq 0$.
\end{proof}

\begin{lemma} \label{lemStriclyGreater}
Let $\rho$ be a density operator and $X$ be some matrix. 
Then ${\rm tr}\{ \rho X \}>0$ implies that there exists a $\ket \psi$ such that $\bra\psi \rho \ket \psi>0$ and $\bra\psi X \ket \psi>0$.
\end{lemma}
\begin{proof}
Let $\rho = \sum_{\lambda_j >0} \lambda_j \ket {u_j}\bra{u_j}$.
Then ${\rm tr}\{ \rho X \} =  \sum_{\lambda_j >0} \lambda_j \bra {u_j} X \ket{u_j}>0$. Hence, there exists a $\ket {u_j}$ for which $\bra{u_j} \rho \ket{u_j}=\lambda_j >0$ and $\bra{u_j} X \ket {u_j}>0$.
\end{proof}

\end{document}